\documentclass[letterpaper, 10pt,conference]{IEEEtran}

\usepackage{amssymb,amsthm}
\usepackage{amsmath}
\usepackage{epsfig}
\usepackage{bm}
\usepackage{verbatim}
\usepackage[left=.78in,top=0.7in,right=.78in,bottom=0.85in,nohead,nofoot]{geometry}

\IEEEoverridecommandlockouts

\newtheorem{lemma}{Lemma}
\newtheorem{proposition}[lemma]{Proposition}

\newtheorem{theorem}[lemma]{Theorem}

\theoremstyle{remark}
\newtheorem{remark}[lemma]{Remark}

\newcommand{\lf}{\lfloor}
\newcommand{\rf}{\rfloor}

\newcommand{\CC}{\mathcal{C}}

\newcommand{\UU}{\mathcal{U}}
\newcommand{\VV}{\mathcal{V}}

\newcommand{\CN}{\mathcal{CN}}

\newcommand{\lam}{\lambda}

\newcommand{\bU}{\mathbf{U}}
\newcommand{\bV}{\mathbf{V}}

\newcommand{\bA}{\mathbf{A}}

\newcommand{\bI}{\mathbf{I}}

\newcommand{\C}{{\mathbb{C}}}

\newcommand{\flr}[1]{\left\lfloor {#1} \right\rfloor}

\newcommand{\comments}[1]{}

\newcommand{\mat}[1]{\left(\begin{matrix}#1\end{matrix} \right)}

\newcommand{\x}{\mathbf{x}}
\newcommand{\y}{\mathbf{y}}
\newcommand{\z}{\mathbf{z}}
\newcommand{\bH}[2]{\mathbf{H}^{[#1 #2]}}
\newcommand{\bHb}[2]{\bar{\mathbf{H}}^{[#1 #2]}}

\newcommand{\herm}{{^*}}

\newcommand{\bAb}{{\bar{ \bA}}}

\newcommand{\im}{\mathrm{Im}}

\author{\IEEEauthorblockN{Guy Bresler}
\IEEEauthorblockA{Wireless Foundations\\
Dept of EECS\\
UC Berkeley\\
Berkeley, CA 94720\\
\textsf{gbresler@eecs.berkeley.edu}}
\and
\IEEEauthorblockN{Dustin Cartwright}
\IEEEauthorblockA{Yale University\\
New Haven, CT 06511\\
\textsf{dustin.cartwright@yale.edu}}
\and
\IEEEauthorblockN{David Tse}
\IEEEauthorblockA{Wireless Foundations\\
Dept of EECS\\
UC Berkeley\\
Berkeley, CA 94720\\
 \textsf{dtse@eecs.berkeley.edu}}}
\comments{
\author{Guy Bresler
\thanks{Guy Bresler is with Wireless Foundations,
Dept of EECS, UC Berkeley, Berkeley, CA 94720,
Email: \textsf{gbresler@eecs.berkeley.edu}.
}
\and
Dustin Cartwright
\thanks{Dustin Cartwright is with
Yale University, New Haven, CT 06511,
Email: \textsf{dustin.cartwright@yale.edu}}
\and
David Tse
\thanks{David Tse is with Wireless Foundations,
Dept of EECS, UC Berkeley, Berkeley, CA 94720,
Email: \textsf{dtse@eecs.berkeley.edu}}}
}

\title{~\newline Geometry of the 3-user MIMO \\ interference channel\thanks{This paper appeared in the \emph{Allerton conference on communication, control, and computing} held September 28-30, 2011}}
\begin{document}
\maketitle
\begin{abstract}
This paper studies vector space interference alignment for the three-user MIMO interference channel with no time or frequency diversity.
The main result is a characterization of the feasibility of interference alignment in the symmetric case where all transmitters have $M$ antennas and all receivers have $N$ antennas. If $N\geq M$ and all users desire $d$ transmit dimensions, then alignment is feasible if and only if $(2r+1)d\leq \max(rN,(r+1)M)$ for all nonnegative integers $r$. The analogous result holds with $M$ and $N$ switched if $M\geq N$. 

It turns out that, just as for the 3-user parallel interference channel \cite{BT09}, the length of alignment paths captures the essence of the problem. In fact, for each feasible value of $M$ and $N$ the maximum alignment path length dictates both the converse and achievability arguments.   

One of the implications of our feasibility criterion is that simply counting equations and comparing to the number of variables \emph{does not predict feasibility}. Instead, a more careful investigation of the geometry of the alignment problem is required. The necessary condition obtained by counting equations is implied by our new feasibility criterion.  
\end{abstract}

\section{Introduction}
Interference alignment has presented the possibility of wildly better performance in interference-limited communications than traditionally thought possible. Introduced by Maddah-Ali et al.
\cite {MMK08} for the multiple-input multiple-output (MIMO) X channel and subsequently by
Cadambe and Jafar \cite{CJ08} in the context of the $K$-user interference
channel (IC), the basic idea is to align multiple interfering signals at each
receiver in order to reduce the effective interference. For the $K$-user IC, in
the case of independently faded parallel channels (i.e. time or frequency
selective), it was shown in \cite{CJ08} that up to $K/2$ total
degrees-of-freedom is achievable: this implies that each user gets
the same degrees of freedom as in a simple 2-user IC, irrespective of the number of users $K$. However, the result
depends critically on the assumption that the number of independently faded
parallel channels, i.e. the channel diversity, is \emph{unbounded} and in fact grows like $K^{2K^2}$. A physical
system has only a \emph{finite} channel diversity, which raises the question of how many degrees of freedom are achievable with the given fixed---finite---amount of diversity. This problem was addressed in the case of time and frequency diversity for 3-user channels in \cite{BT09}.

In this paper we consider the $K$-user MIMO IC, where each transmitter and receiver has multiple antennas, but the channel is constant over time and frequency. Similar in flavor to the situation with finite time or frequency diversity in \cite{BT09}, here we have a fixed amount of spatial diversity due to the multiple antenna elements, and the goal is to design the best communication strategy (in terms of degrees-of-freedom) for the system at hand.  
 Finding the maximum degrees-of-freedom available is equivalent to the problem of determining, for a fixed number of degrees-of-freedom $d$ per user, the region of $M$ and $N$ for which there exists a valid linear encoding and decoding strategy (as defined in Subsection~\ref{ss:DOF}). This objective is accomplished in the following Theorem.

\begin{figure}[tb]
\begin{center}
\includegraphics[width=3.65in]{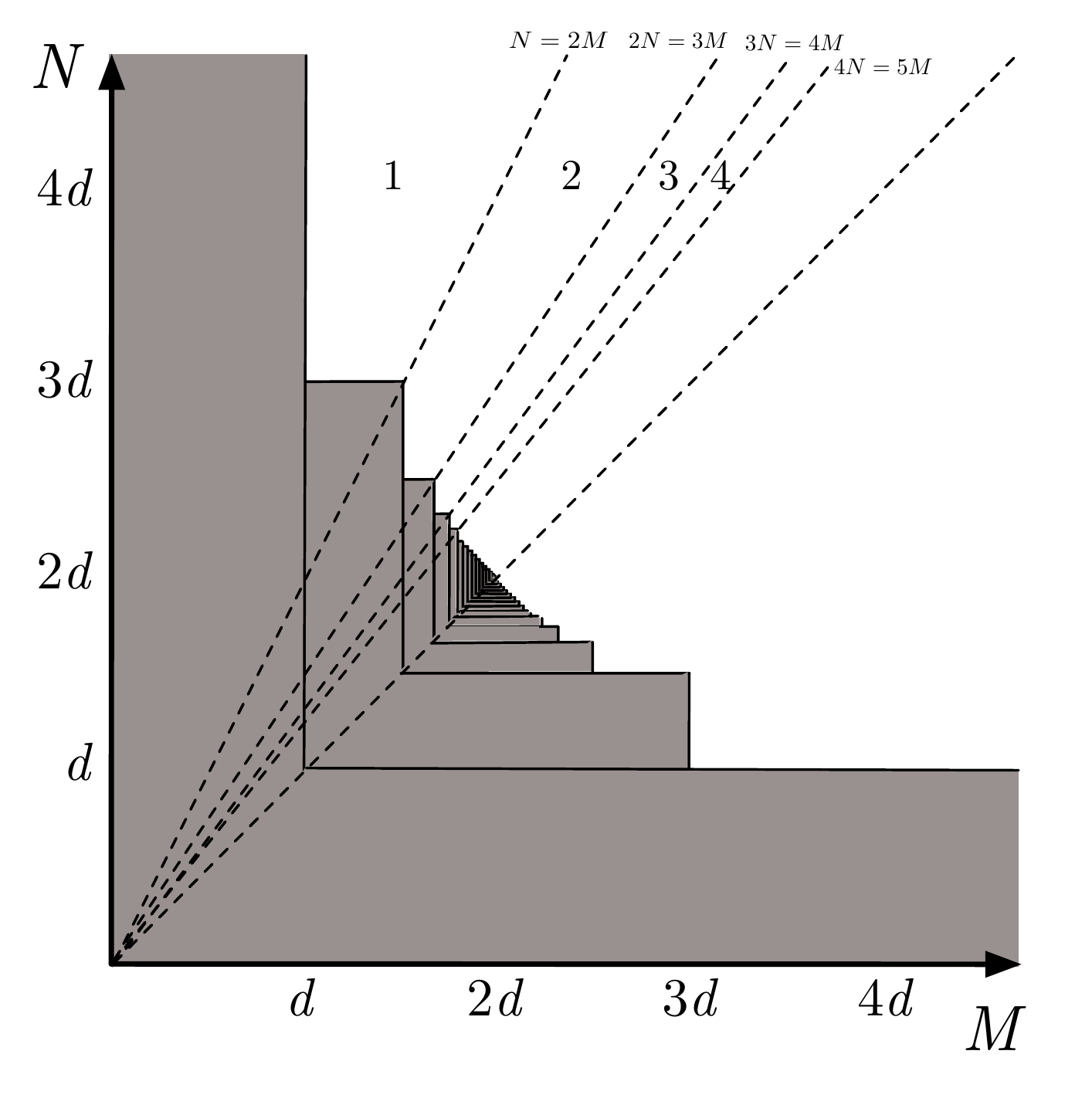} 
\vspace{-0.9cm}
\caption{For a fixed value of $d$, the feasible region in the $M,N$ plane is white while the infeasible region is shaded. 
The labels $1,2,3,4,...$ indicate the maximum length of alignment paths for $M,N$ in the corresponding region.} 
\label{f:feasibleRegion}
\end{center}
\vspace{-.6cm}
\end{figure}

\begin{theorem}[Symmetric degrees-of-freedom] \label{t:SymDOF}
Consider the three-user Gaussian MIMO interference channel with no time or frequency diversity. Fix the number of desired transmit dimensions $d_i=d$, transmit antennas $M_i=M$, and receive antennas $N_i=N$. Assume without loss of generality that $N\geq M$. Then alignment is feasible if and only if
\begin{equation}\label{e:feasibility} 
(2r+1)d\leq \max(rN,(r+1)M),\quad \text{for all integers } r\geq 0\,.
\end{equation}
\end{theorem}

\subsection{Explanation of result}
Before providing rigorous arguments, in this subsection we describe the geometry underlying the problem. Relevant definitions are provided in subsequent sections. 

A given vector $u_i$ in the signal space of transmitter $i$ is said to initiate an alignment path of length $r+1$ if there exists a sequence of vectors $u_{i+1},u_{i+2},\dots,u_{i+r}\in \C^M$, such that $\bH {i-1,}{i} u_i = \bH{i-1,}{i+1}u_{i+1}, \dots, \bH {i+r-2,}{i+r-1} u_{i+r-1} = \bH{i+r-2,}{i+r} u_{i+r}$. Here channel indices are interpreted modulo~3. For example, a vector $u_2$ at transmitter 2 initiating an alignment path of length~3 means that there exist vectors $u_3$ and $u_1$ such that $\bH12 u_2 = \bH13 u_3$ and $\bH23 u_3 = \bH21 u_1$.

The feasible region of Figure~\ref{f:feasibleRegion} is divided up into sub-regions labeled with the maximum length of an alignment path; this number depends on $M$ and $N$ through the incidence geometry of the images of the channel matrices~$\im(\bH ij)$.
We begin by examining sub-region~1, and then look at how things generalize to the other sub-regions. 

The point of departure is the obvious constraint $d\leq M$ in order to have a $d$-dimensional subspace of an $M$ dimensional vector space. Continuing, assuming $M\geq d$, suppose $ 2M\leq N$, so $(M,N)$ lies in sub-region~1 of Figure~\ref{f:feasibleRegion}. At receiver one, the images $\im(\bH12)$ and $\im(\bH13)$ of the channels from transmitters two and three are in general position and therefore their intersection has dimension $[2M-N]^+=0$; in other words, \emph{alignment is impossible} in sub-region~1. Figure~\ref{f:NoAlignment} shows pictorially that because alignment is not possible here, we have the constraint $3d\leq N$. Mathematically, we see that alignment is not possible because the map from $\C^{2M}$ to $\C^N$ given by the matrix $\mat{\bH12 &\bH13 }$ is injective. 
\begin{figure}[bt]
\begin{center}
\includegraphics[width=3.5in]{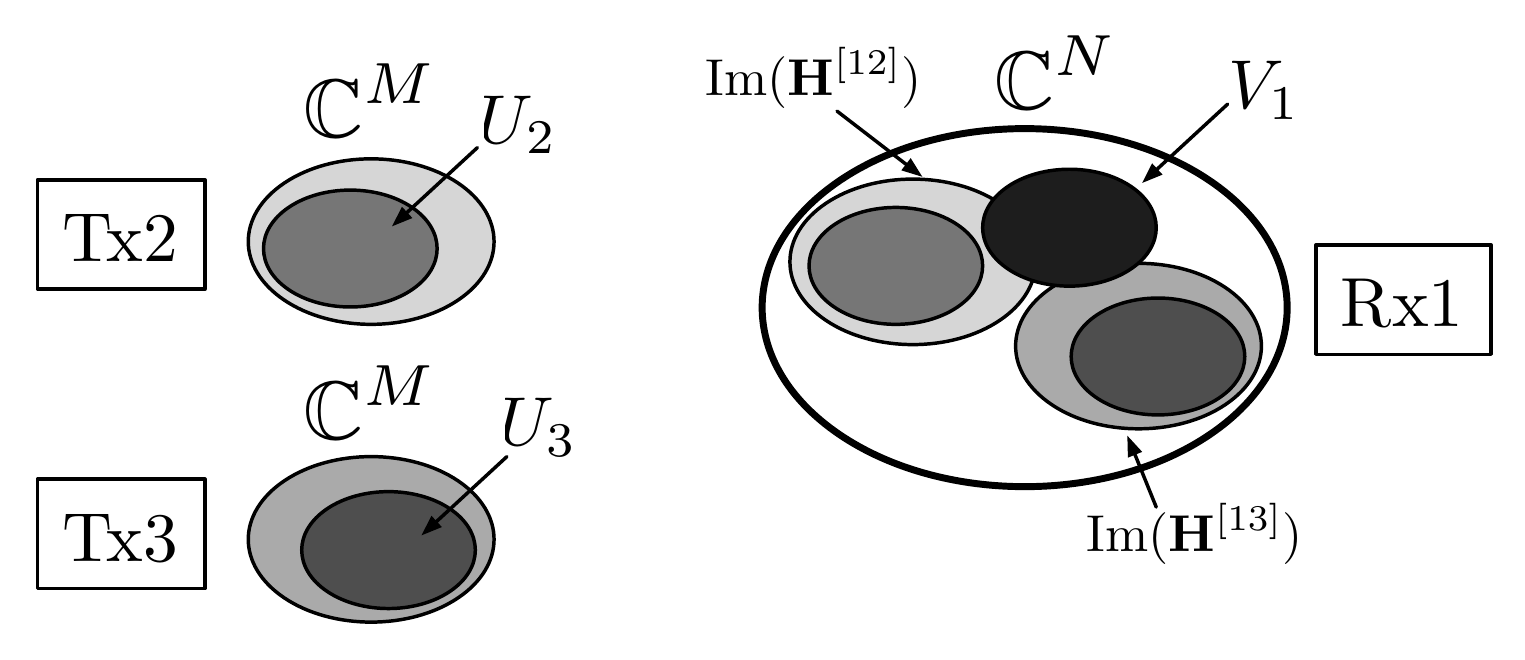}
\caption{Sub-region~1: The figure indicates that no alignment is possible when $2M\leq N$, since $\im(\bH12)$ and $\im(\bH13)$ are complementary. Since the three subspaces $V_1, \bH12 U_2,\bH13 U_3$ are each of dimension $d$, complementary, and lie in $\C^N$ at receiver 1, we obtain the constraint $3d\leq N$.  } 
\label{f:NoAlignment}
\end{center}
\vspace{-.3cm}
\end{figure}

Moving onward to sub-region~2, we have $2M> N$ and thus alignment \emph{is} possible. This means that alignment paths of length 2 are possible (Fig~\ref{f:AlignmentPath2}), with up to $2M-N$ interference dimensions overlapping at each receiver. 
Thus, the interference space $\bH12 U_2 + \bH13 U_3$ at receiver one occupies at least $2d-(2M-N)$ dimensions, and we have the constraint $3d\leq 2M$.
However, because $3M\leq 2N$, no vector at (say) transmitter three can be \emph{simultaneously} aligned at both receivers one and two, as indicated in Figure~\ref{f:NoSimulAlignment}. One can also see that no simultaneous alignment is possible by changing change perspective to that of a combined receiver one and two. 
 One may check directly that (as a special case of Lemma~\ref{l:AfullRank}), the map 
\begin{equation}\label{e:Aexample}
\mat{\bH12 & \bH13 & \\ & \bH23 & \bH21}
\end{equation}
from the three transmitters to $\C^{2N}$ is injective; analogously to the case in sub-region 1, this is interpreted to mean that no alignment is possible in the combined receive space $\C^{2N}$ (see Fig.~\ref{f:NoSimulAlignmentBound}). Thus, five complementary $d$-dimensional  subspaces lie in $\C^{2N}$ and we obtain the constraint $5d\leq 2N$. 

As far as achievability goes, the basic rule-of-thumb is to create alignment paths of maximum length. Thus, in sub-region~2, where alignment is possible, the achievable strategy aligns as many vectors as possible and the remaining ones (if $d>2(2M-N)$) are not aligned.

\begin{figure}[tb]
\begin{center}
\includegraphics[width=3.5in]{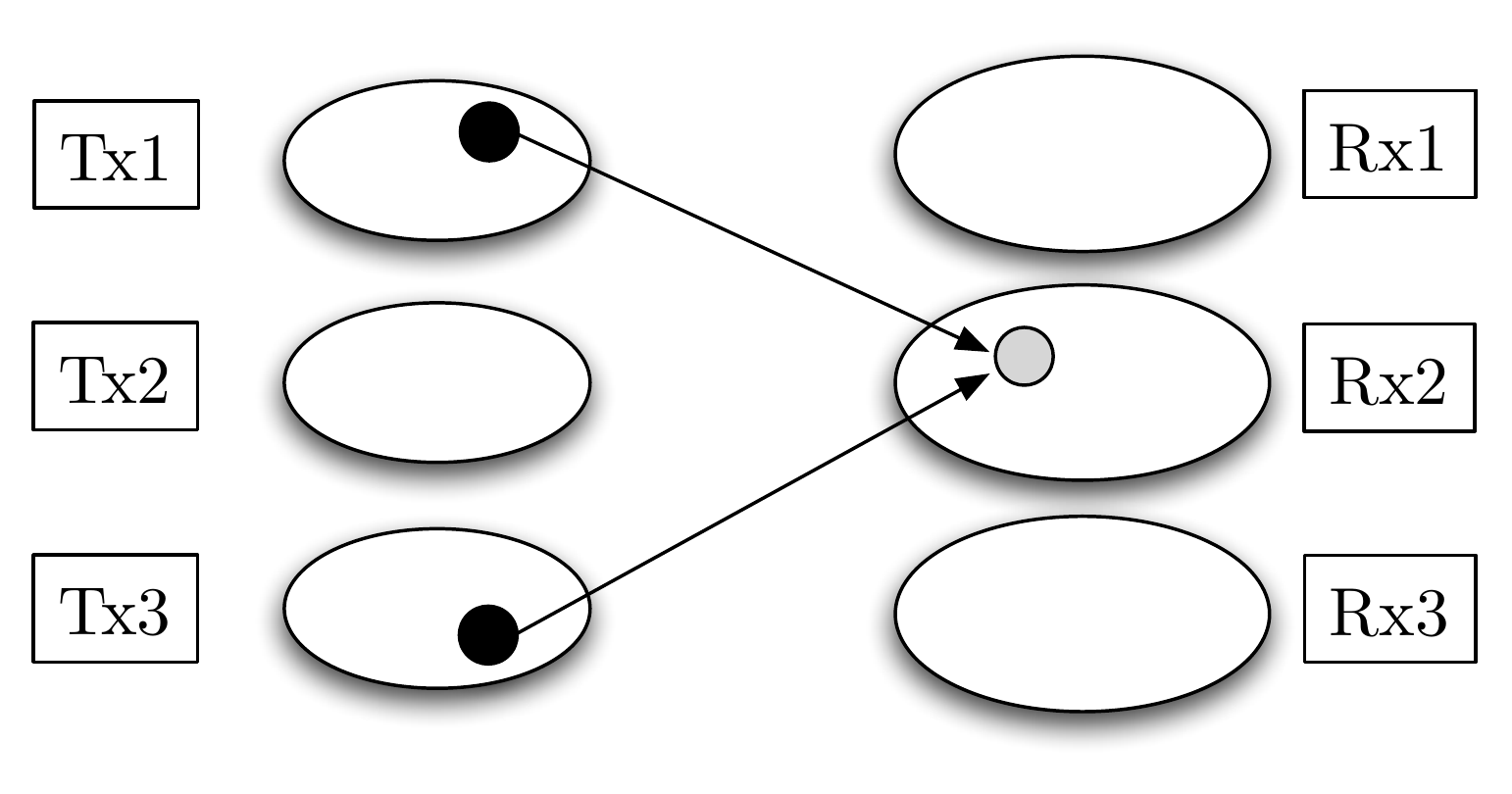}
\caption{Sub-region~2: Alignment is possible here. The figure denotes an alignment path of length~2.} 
\label{f:AlignmentPath2}
\end{center}
\end{figure}

\begin{figure}[tb]
\begin{center}
\includegraphics[width=3.5in]{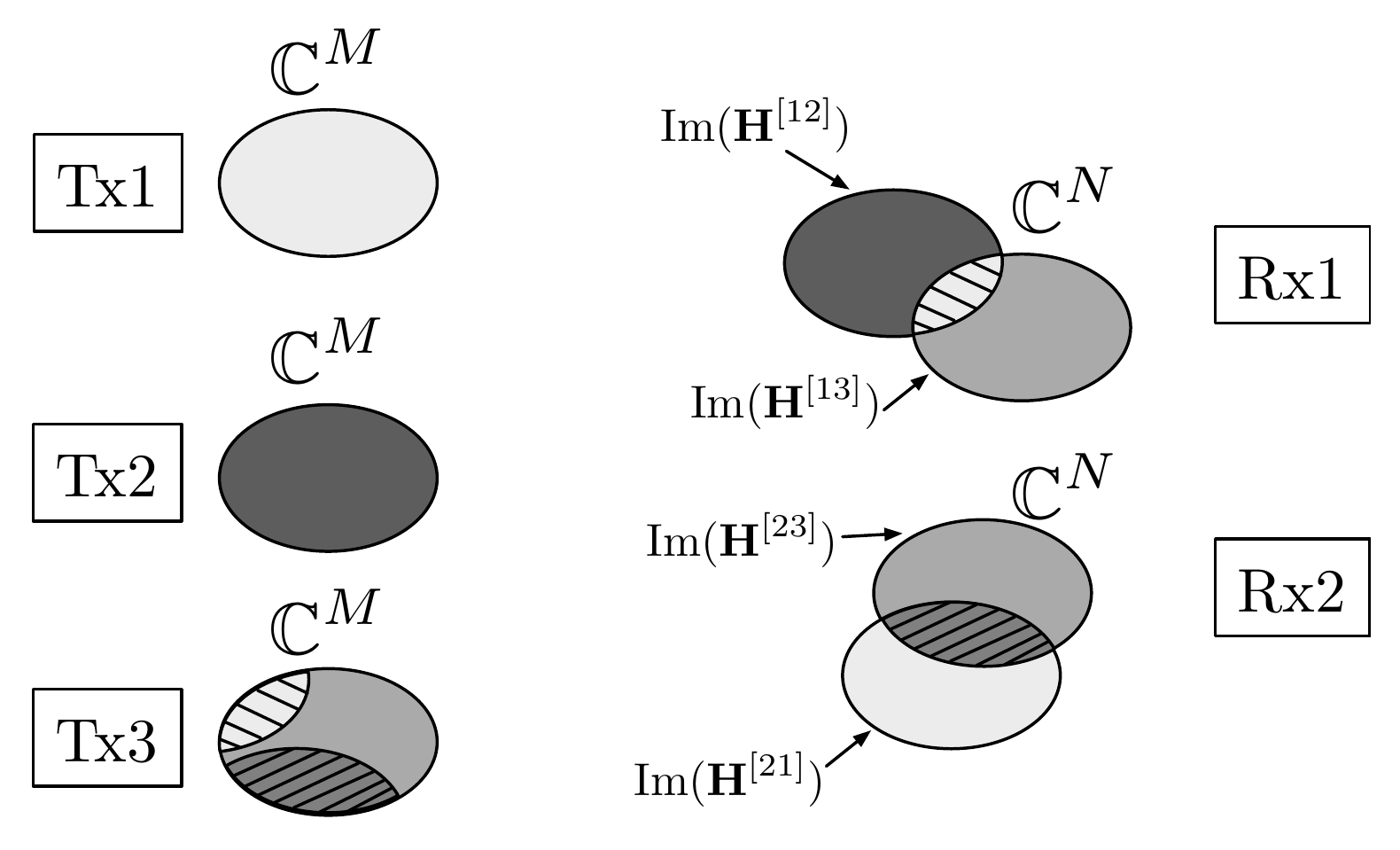}
\caption{Sub-region~2: The striped regions at receivers one and two each denote the dimension $2M-N$ portion of the space in which alignment can occur. From transmitter three's perspective, one sees that  \emph{simultaneous} alignment is not possible for $2(2M-N)\leq M$, or equivalently, $3M\leq 2N$. } 
\label{f:NoSimulAlignment}
\end{center}
\vspace{-.3cm}
\end{figure}

\begin{figure}[tb]
\begin{center}
\includegraphics[width=3.5in]{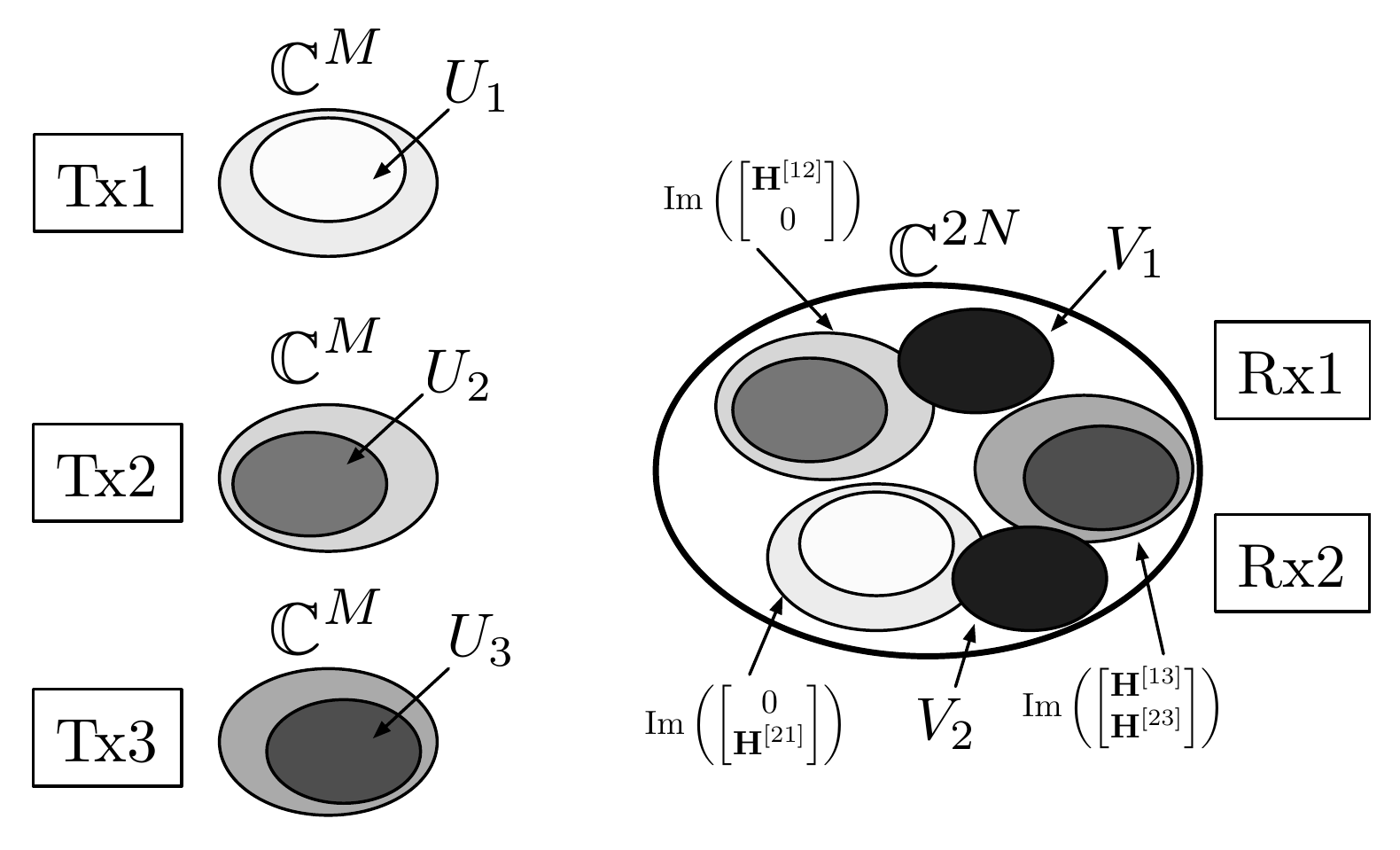}
\caption{Sub-region~2: Considering the dimension $2N$ receive space formed by receivers one and two together, along with the map defined \eqref{e:Aexample} from the three transmitters, shows that no alignment is possible in this combined space. Since there are five complementary subspaces of dimension $d$ we obtain the constraint $5d\leq 2N$. }
\label{f:NoSimulAlignmentBound}
\end{center}
\vspace{-0.5cm}
\end{figure}

Both the necessary conditions and achievability arguments extend in a natural way. On the achievability end, alignment paths of maximum length are used. For example, in sub-region~4, alignment paths of length four are used (Fig~\ref{f:AlignmentPath4}). For the converse, a generalization of the matrix in \eqref{e:Aexample} is shown to be full-rank in Lemma~\ref{l:AfullRank}, giving the constraints in \eqref{e:feasibility}.
\begin{figure}[tb]
\begin{center}
\includegraphics[width=3.5in]{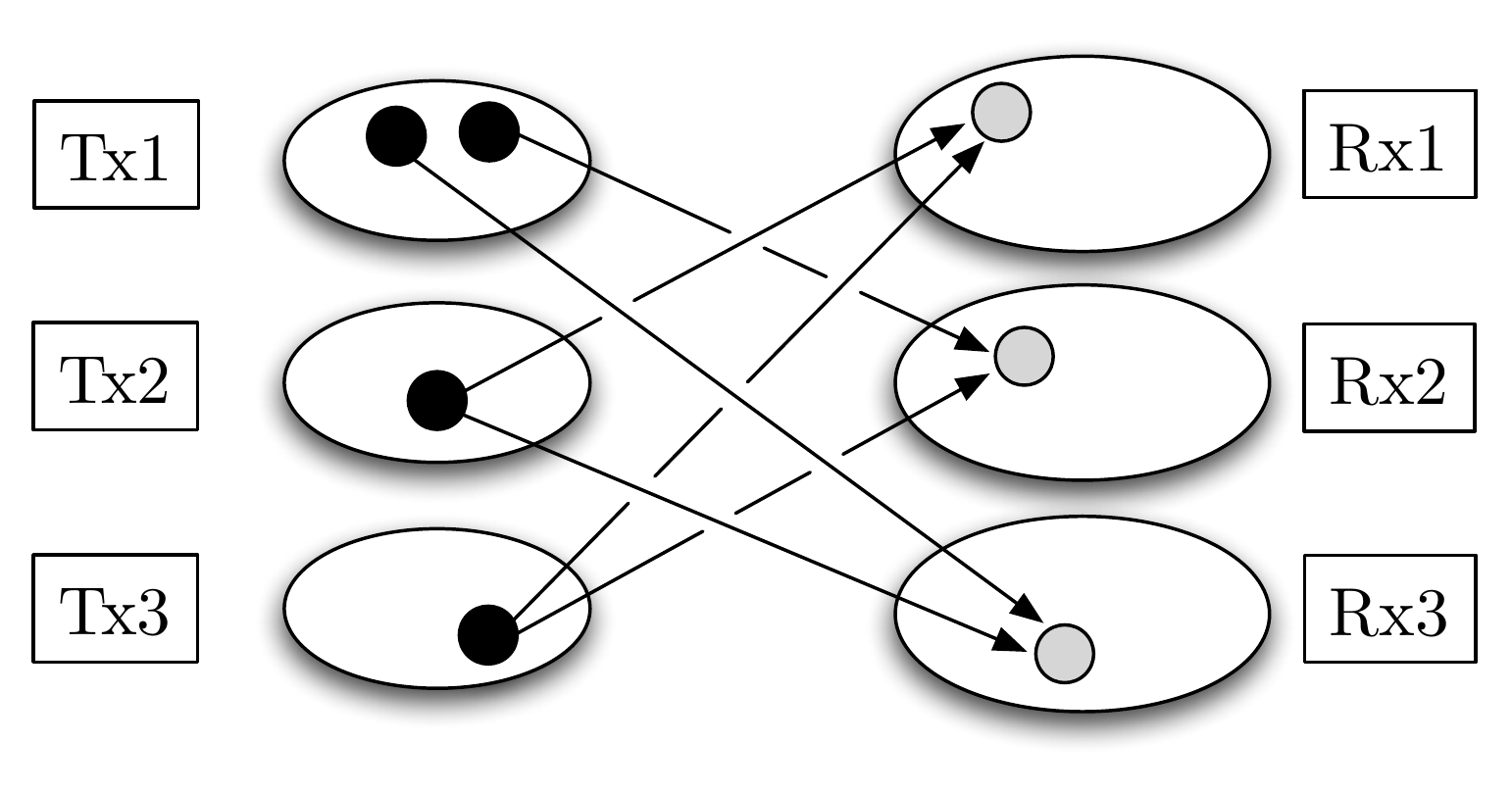}
\caption{Sub-region~4: Alignment paths of length four are denoted here, initiated by vectors at transmitter~1.} 
\label{f:AlignmentPath4}
\end{center}
\vspace{-0.5cm}
\end{figure}

\subsection{Related work}

The problem we consider, of maximizing degrees-of-freedom using linear strategies for the $K$-user MIMO IC with finite number of transmit and receive antennas, has received significant attention in the last several years. Jafar and Fakhereddin \cite{JF07} determined the degrees of freedom of the two-user MIMO IC with an arbitrary number of antennas at each of the four terminals. 
Cadambe and Jafar \cite{CJ08} considered the problem for $K=3$ users and $N=2$ antennas, and showed that $3/2$ dof was achievable. For more than $3$ users or $N>2$ they assumed infinite time or frequency diversity and applied their main $K/2$ result. Gomadam et al. \cite{GCJ08,GCJ11}, posed the problem of determining feasibility of linear alignment in the constant channel setting, but left the problem unanswered and proposed a heuristic iterative numerical algorithm.

Also at the heuristic level, Yetis et al. \cite{YGJK10} proposed to determine feasibility of alignment by counting the number of equations and comparing to the number of variables. This approach was carried out rigorously to show a necessary and sufficient condition in \cite{BCT11} for the symmetric square case of $N=M$ antennas at all $K$ transmitters and receivers and by Razaviyayn et al. \cite{RGL11} for the case that the number of transmit dimensions $d$ divides both the number of transmit and receive antennas. Several other works have subsequently pursued a similar approach for related problems, including~\cite{ZBH11}, \cite{GG11} (both heuristic), and \cite{AEV11}. 


The remaining work on linear alignment for the MIMO IC has focused on heuristic algorithms, mainly
iterative in nature (see~\cite{GCJ11}, \cite{PH09}, \cite{RSL10}, \cite{PD10}, and \cite{SGHP10}). Some have proofs of convergence, but no performance guarantees are known.

We emphasize that in this paper we restrict attention to vector space interference alignment,
where the effect of finite channel diversity can be observed. Interfering signals can also be aligned on the signal scale
using lattice codes (first proposed in~\cite{BPT10}, see also \cite{CJS09}, \cite{EO09}, \cite{MGMK09}), however the understanding of this type of alignment is currently at the stage corresponding to infinite time or frequency diversity in the vector space setting. In other words, essentially ``perfect" alignment is possible due to the infinite channel precision available at infinite signal-to-noise ratios. 

Ghasemi et al. \cite{GMK10} apply alignment on the signal scale to the $K$-user $M\times N$ MIMO IC. The converse arguments in that paper are obtained by forming a two-user interference channel with two users transmitting and decoding jointly; thus they obtain the inequality $3d\leq \max(N,2M)$ corresponding to $r=1$ in \eqref{e:feasibility} of the present paper. 

In exactly the same setting as the present paper, Amir et al. \cite{AEN11} have independently proposed a similar achievable strategy for critical $M,N$ satisfying both \eqref{e:feasibility} and $M+N=4d$. \cite{AEN11} is limited to critical values of $M,N$ and contains no converse arguments beyond the equation counting bound of \cite{BCT11} and \cite{RGL11}.

Finally, also independently, Wang et al. \cite{WGJ11} very recently posted a paper to the Arxiv containing many similar results. Their converse is information theoretic and, unlike ours, is not limited to linear strategies.

\subsection{Interference channel model}
The $K$-user MIMO interference channel has $K$ transmitters and $K$ receivers, with transmitter $i$ having $M_i$ antennas and receiver $i$ having $N_i$ antennas. For $i=1,\dots,K$, receiver $i$ wishes to obtain a message from the corresponding transmitter $i$. The remaining signals from transmitters $j\neq i$ are undesired interference.
The channel is assumed to be constant over time, and at each time-step the input-output relationship is given by
\begin{equation}
  \y_i=\bH{i}{i}\x_i+\sum_{1\leq j\leq K\atop j\neq i}\bH{i}{j}\x_j+\z_i\,,\quad 1\leq i\leq K\,.
\end{equation}
Here for each user $i$ we have $\x_i\in \C^{M_i}$ and $\y_i, \z_i\in \C^{N_i}$, with $\x_i$ the transmitted signal, $\y_i$ the received signal,  and $\z_i\sim \CN(0,I_{N_i})$ is additive isotropic white Gaussian noise. The channel matrices are given by $\bH{i}{j}\in \C^{N_i\times M_j}$ for $1\leq i,j\leq K$, with each entry assumed to be independent and with a continuous distribution. We note that this last assumption on independence can be weakened significantly to a basic non-degeneracy condition but we will not pursue this here. For our purposes this means the channel matrices are generic. Each user has an average power constraint, $E(||\x_i||^2)\leq P$.

\subsection{Vector space strategies and degrees-of-freedom}\label{ss:DOF}
We restrict the class of coding strategies to (linear) \emph{vector space} strategies.
In this context degrees-of-freedom (dof) has a simple interpretation as the dimensions of the transmit subspaces, described in the next paragraph. However, note that one can more generally define the degrees-of-freedom region in terms of an appropriate high transmit-power limit $P\to \infty$ of the Shannon capacity region $C(P)$ normalized by $\log P$ (\cite{CJ08}, \cite{MMK08}). In that general framework, it is well-known and easy to show that vector space strategies give a concrete non-optimal achievable strategy with
rates $$R_i(P)=d_i\log(P)+O(1), \quad 1\leq i\leq K\,.$$ Here $d_i$ is the dimension of transmitter $i$'s subspace and $P$ is the transmit power.

The transmitters encode their data using vector space precoding.
Suppose transmitter $j$ wishes to transmit a vector $\hat x_j \in \mathbb
C^{d_j}$ of
$d_j$ data symbols. These data symbols are modulated on the subspace
$U_j\subseteq \C^{M_j}$ of dimension~$d_j$, giving the input signal
$\x_j= \bU_j\hat x_j$, where $ \bU_j$ is a $M_j \times d_j$ matrix whose
column span is $U_j$. The signal~$\x_j$ is received by receiver $i$ through the channel
as $\bH{i}{j} \bU_j\hat x_j$ 
The dimension of the transmit
space, $d_j$, determines the number of data streams, or degrees-of-freedom,
available to transmitter $j$. With this restriction of strategies, the output is
given by
\begin{equation}
  \y_i=\bH{i}{i} \bU_i\hat x_i+\sum_{1\leq j\leq K\atop j\neq i}\bH{i}{j}\tilde \bU_j\hat \x_j+\z_i\,,\quad 1\leq i\leq K\,.
\end{equation}
The desired signal space at receiver $i$ is thus $\bH{i}{i}U_i$, while the
interference space is given by $\sum_{j\neq i}\bH{i}{j}U_j$, i.e.\ the span of the undesired subspaces as observed by receiver~$i$.

In the regime of asymptotically high transmit powers, in order that decoding can be accomplished we impose the constraint at each receiver~$i$ that the desired signal space $\bH{i}{i}U_i$ is complementary to the interference space $\sum_{j\neq i}\bH{i}{j}U_j$.  Equivalently, there must exist subspaces $V_i$ with $\dim V_i=
\dim U_i$ such that
\begin{equation}\label{e:Generalorthogonality}
 \bH{i}{j}U_j \perp V_i\,, \quad 1\leq i,j \leq K, \quad i\neq j\,,
\end{equation}
 and
\begin{equation}\label{e:projCond}
\dim(\text{Proj}_{V_i}\bH{i}{i}U_i)=\dim U_i\,.
\end{equation}
Here $\bH{i}{j}U_j\perp V_i$ is interpreted to mean that $V_i$ belongs to the dual space $(\C^{N_i})^*$ and $V_i$ annihilates $\bH{i}{j}U_j$. Alternatively, $( \bV_i)^\dagger \bH{i}{j} \bU_j=0$, where $\bV^\dagger$ denotes the Hermitian transpose of $\bV$ and $ \bV_i$ is a matrix with column span equal to $V_i$. Note that implicitly the transmit dimensions are assumed to satisfy the obvious inequality $d_i\leq \min(M_i,N_i)$.
 If each direct channel matrix $\bH{i}{i}$ has generic (or i.i.d. continuously distributed) entries, then the second condition~\eqref{e:projCond} is satisfied assuming $\dim V_i=d_i$ for each $i$ (this can be easily justified---see \cite{GCJ08} for some brief remarks). Hence we focus on condition \eqref{e:Generalorthogonality}.

The goal is to maximize degrees of freedom, i.e.
choose subspaces $U_1,\dots,U_K$, $V_1,\dots,V_K$ with $d_i\leq\min(M_i,N_i)$ in order to
\begin{align*}
&\text{maximize}\quad d_1+d_2+\dots+d_K \\
&\text{subject to}\quad  \bH{i}{j}U_j \perp V_i\,, \quad 1\leq i,j \leq K, \quad i\neq j\,,
\end{align*}
To this end, it is sufficient to answer the following feasibility question: given number of users $K$, number of antennas $M_1,\dots,M_K$, $N_1,\dots,N_K$, and desired transmit subspace dimensions $d_1,\dots,d_K$, does there exist a choice of subspaces $U_1,\dots,U_K$ and $V_1,\dots,V_K$ with $\dim U_i=\dim V_i=d_i,1\leq i\leq K$, satisfying \eqref{e:Generalorthogonality}?

\comments{
One requirement for achievability is that the number of parameters must be less
than the dimension of the constraints imposed by \eqref{e:Generalorthogonality}. See \cite{BCT11}
and~\cite{RGL11}, where this argument is carried out to obtain necessary
conditions on the feasibility of interference alignment. In this paper, we
obtain for the first time constraints which cannot be obtained from such
dimension counts.
}

For the rest of the paper, we fix $K=3$, and we introduce two notational
conveniences for this case.  First, we interpret the indices modulo three, so that $\bH12 = \bH42$ and so
on. Second, since the indices can always be chosen to differ by exactly
one, we will adopt the shorthand $\bH{i,}{+}$ and $\bH{i,}{-}$ for
$\bH{i,}{i+1}$ and $\bH{i,}{i-1}$ respectively.


\section{Proof of converse}

We begin with a key lemma.
\begin{lemma}\label{l:AfullRank}
Suppose $N\geq M$. For any $r\geq 1$ define the $rN\times (r+1)M$ block matrix 
\begin{equation}\label{e:big-matrix}
\bA_r=\mat{
\bH12 & \bH13 & & & &
\\
 & \bH23 & \bH21 & & &
 \\
& & \bH31 & \bH32 & &
\\
& & & \ddots & &
\\
& & & & \bH{r,}{+} & \bH{r,}{-}
}\,,
\end{equation}
where the indices are interpreted as described above. For generic channel matrices $\bH
ij$, the matrix $\bA_r$ has full rank. 
\end{lemma}

\begin{proof}
In order to prove that $\bA_r$ has full rank for generic channel matrices, it is
sufficient to prove that it does for one particular set of matrices (see e.g. \cite{Hartshorne}).
We specialize to the matrices
\begin{equation*}
B := \bH{1,}{+} = \bH{2,}{+} = \bH{3,}{+}
= \begin{pmatrix} I_{M} \\ \mathbf 0 \end{pmatrix}
\end{equation*}
and 
\begin{equation*} C := \bH{1,}{-} = \bH{2,}{-} = \bH{3,}{-}
= \begin{pmatrix} \mathbf 0 \\ I_{M} \end{pmatrix},
\end{equation*}
where $I_M$ denotes the $M\times M$ identity matrix and the $\mathbf 0$ denotes
a block of $0$s of size $(N-M) \times M$.

We will prove that, with these specializations, $\bA_n$ has full rank by
simultaneous induction on~$r$, $N$, and~$M$. If $r=0$, then $\bA_r$ is a $0
\times M$ matrix, which trivially has full rank. If $N \geq 2M$, then every row
vector is a unit vector and all such unit vectors appear in some row, so the
matrix has full rank.

Now we suppose that $N < 2M$.
We permute the rows and columns of $\bA_n$ as follows. We extract the first
block of $N$ rows, followed by the last $N-M$ rows of each of the subsequent
$r-1$ blocks. We put these rows last, after the remaining rows, each in
their induced order. Similarly, we take the first block of $M$~columns, followed
by the last $N-M$ columns of each of the other $r$ column blocks, and place
these to the right of all the other columns. This creates a square matrix of
size $M + r(N-M)$ in the lower right, and we will prove that this submatrix is
the identity.

If we divide $B$ and~$C$ into blocks by separating off the last $N-M$ rows and
columns of each, then using our assumption that $N-M < M$, we get
\begin{equation*}
B = \begin{pmatrix} \tilde B & B' \\ \mathbf 0 & \mathbf 0 \end{pmatrix}
\quad
C = \begin{pmatrix} \tilde C & \mathbf 0 \\ \mathbf 0 & I_{N-M} \end{pmatrix}\,,
\end{equation*}
where
\begin{equation*}B' = \begin{pmatrix} \mathbf 0 \\ I_{N-M} \end{pmatrix}
\quad
\tilde B = \begin{pmatrix} I_{2M-N} \\ \mathbf 0 \end{pmatrix}
\quad
\tilde C = \begin{pmatrix} \mathbf 0 \\ I_{2M-N} \end{pmatrix}\,.
\end{equation*}
Therefore, the rearranged matrix has the form
\begin{equation*}
\begin{pmatrix}
\tilde B&\tilde C&        &        &        &B'  \\
        &\ddots  &\ddots  &        &        &        & \ddots \\
        &        &\tilde B&\tilde C&        &        &       & B' \\
\tilde C&        &        &        & I_{M}  &        &  \\
        &        &        &        &\mathbf0&I_{N-M} &  \\
        &        &        &        &        &\mathbf0&I_{N-M}& \\
        &        &        &        &        &        &\ddots & \ddots \\ 
        &        &        &        &        &        &       &\mathbf 0 & I_{N-M} \\
\end{pmatrix}.
\end{equation*}
We can use the central $I_M$, together with elementary column operations to
clear the $\tilde C$ on the left. Similarly, elementary row operations can use
the diagonal $I_{N-M}$ to clear the $B'$s in the upper right. The only remaining
non-zero entries are in the $(M+r(N-M)) \times (M+r(N-M))$ identity matrix in
the lower right and the upper left block, with the copies of~$\tilde B$
and~$\tilde C$.
The latter matrix is just our specialized version of~$\bA_{r-1}$ with parameters
$M$ and~$N$ each decreased by $N-M$, and this matrix has full rank by the
inductive hypothesis.
\end{proof}

\comments{
An analogous lemma also holds for time-extended channels.

\begin{lemma}\label{l:AxTfullRank}
Suppose $N\geq M$ and let $n$ and $T$ be a positive integers. Define the $nNT\times (n+1)MT$ block matrix 
$$
\bAb_n=\mat{
\bHb12 & \bHb13 & & & &
\\
 & \bHb23 & \bHb21 & & &
 \\
& & \bHb31 & \bHb32 & &
\\
& & & \ddots & &
\\
& & & & \bHb{n,}{+} & \bHb{n,}{-}
}\,,
$$
where the indices are interpreted modulo $3$. Here $\bHb ij= \bH ij\otimes \bI_T$. For generic channel matrices $\bH ij$, the matrix $\bA_n$ has full rank. 
\end{lemma}
\begin{proof}
By rearranging rows and columns of $\bAb_n$ one may obtain the matrix $\bA_n\otimes \bI_T$,  which by  Lemma~\ref{l:AfullRank} implies that $\bAb_n$ generically has full rank. \end{proof}
}

The following proposition uses the preceding lemma to prove a new set of constraints.

\begin{proposition}[Converse]\label{p:necessary}
Suppose $N\geq M$ are integers. Fix the degrees of freedom per user $d_i=d$ and number of antennas $M_i=M,N_i=N$.  Alignment is feasible only if
$$
(2r+1)d\leq \max(rN,(r+1)M),\quad \text{for all } r\geq 0\,.
$$
\end{proposition}

\begin{remark}
Proposition~\ref{p:necessary} remains valid when allowing constant channel time extension, with $M,N,$ and $d$ appropriately normalized by the time extension value.
\end{remark}
\begin{proof} 
We fix the value of $r\geq0$, and omit dependence on $r$ whenever convenient.
Define the product of transmit spaces  $\UU=U_2\times U_3\times \dots \times
U_{r+2}\subset(\C^M)^{r+1}$, where as usual indices are interpreted modulo~$3$,
and similarly let $\VV=V_1\times\dots V_r\subset (\C^N)^r$. Note that each $U_i$
and $V_i$ has dimension $d$, so $\UU$ and~$\VV$ have dimensions~$(r+1)d$
and~$rd$ respectively.

First, suppose that $rN \geq (r+1)M$. Then Lemma~\ref{l:AfullRank} implies that
the linear map $\bA_r\colon (\C^M)^{r+1}\to(\C^N)^r$ is injective. By the
orthogonality condition~\eqref{e:Generalorthogonality},
we obtain $\VV\perp \bA_r \UU$, and thus $rd+(r+1)d=\dim \VV +
 \dim(\bA_r\UU) \leq\dim (\C^N)^r = rN$.

On the other hand, if $(r+1)M\geq rN$, the Hermitian transpose $\bA_r\herm$ is
an injective linear map $\bA_r\herm\colon (\C^N)^r\to(\C^M)^{r+1}$. Again, the
orthogonality conditions~\eqref{e:Generalorthogonality} imply that $\bA_r\herm\VV
\perp \UU$ so $(2r+1)d\leq (r+1)M$. This proves the lemma.
\end{proof}

Note that when $r=0$ Proposition~\ref{p:necessary} reduces to the obvious
constraint $d \leq M$ in order to have a $d$-dimensional subspace of an
$M$-dimensional vector space. In fact, the proposition and its proof can be
considered generalizations of this observation, with the inequality arising from
the fact that the vector spaces $\VV + \bA_r\UU$ or $\bA_r\herm \VV + \UU$ must be
contained in $(\C^N)^r$ and $(\C^M)^r$ respectively.


\section{Proof of Achievability}

\begin{theorem}[Achievability]\label{t:achievability}
Fix any $M,N,$ and $d$ satisfying \eqref{e:feasibility}. Then alignment is feasible, i.e. there exists a choice of subspaces $U_1,U_2,U_3, V_1,V_2,V_3$ with $\dim U_i=\dim V_i=d$, for $1\leq i\leq 3,$ and $V_i\perp \bH ij U_j$ for~$1\leq i\neq j\leq 3$. 
\end{theorem}
\begin{proof}
The proof for the critical points satisfying $N+M=4d$ is given as part of Proposition~\ref{p:critical} below. The more general argument is similar, but tedious, and deferred to the appendix.   
\end{proof}

\begin{remark}
The achievable strategy specifies an explicit construction for the solutions in terms of the kernel of an appropriate matrix (or in terms of eigenvectors in the case $M=N$). This contrasts with the existence proofs for $K>3$ in \cite{BCT11} and \cite{RGL11}, which do not provide a way to find solutions.
\end{remark}

\begin{proposition}\label{p:critical}
Fix integers $d$ and $N\geq M$ satisfying $N+M=4d$. Then alignment is
feasible if and only if either $N=M = 2d$ or the integer $d$ is evenly divisible
by $2d-M = N-2d$.
\end{proposition}

\begin{proof}
The necessity follows by some manipulations of Proposition~\ref{p:necessary}. If
$N \neq M$ and $d/(2d-M)$ is not an integer, then we set $r$ to be the nearest
integer to $M/(N-M)$, which is well-defined because of the equality:
\begin{equation*}
\frac{M}{N-M} = \frac{d}{2d-M} - \frac{1}{2}.
\end{equation*}
Thus,
\begin{equation*}
r = \frac{M}{N-M} + e
\end{equation*}
where $e$ has absolute value strictly less than one half. Now, we get
\begin{align*}
(2r+1)d &= \frac{(N+M)^2}{4(N-M)} + \frac{e(N+M)}{2} \\
rN &= \frac{NM}{N-M} + e N \\
(r+1)M &= \frac{NM}{N-M} + e M.
\end{align*}
Which of the latter two is larger will depend on the sign of $e$. Assuming that
$e$ is positive, we can substitute and clear denominators to get that
\begin{equation*}
(2r+1)d \leq \max\{rN, (r+1)M\}
\end{equation*}
is equivalent to
\begin{align*}
0 &\geq (N+M)^2 + 2e(N+M)(N-M) \\&\quad - 4NM - 4eN(N-M) \\&= 
(N-M)^2 - 2e(N-M)^2,
\end{align*}
which will be false because $e$ is less than one half. The case when $e$ is
negative works similarly.

We now turn to the sufficiency part of the proof. Suppose that $2d-M$ is positive and evenly divides $d$.
We set $r = d/(2d-M) -1$, from which it follows that $M= d(2r+1)/(r+1)$ and
$N=d(2r+3)/(r+1)$.
For any integer $i$, we define shifted versions of the block matrix
from~(\ref{e:big-matrix}):
\begin{align*}
&\bA_r^i=\\&\mat{
\bH{i,}+ &\! \!\bH{i,}- & & & &
\\
 & \! \!\bH{i+1,}+ & \! \! \bH{i+1,}{-} & & &
 \\
& & & \!\!\ddots & &
\\
& & & & \!\!\bH{i+r-1,}{+} &\!\! \bH{i+r-1,}{-}
}
\end{align*}
By Lemma~\ref{l:AfullRank}, for generic channel matrices, $\bA_r^i$
has full rank.
Therefore, its kernel is a vector space of dimension $(r+1)M-rN = d/(r+1)$, and
we denote this vector space by $W_i$. For $i+1 \leq j \leq i+r+1$, define
$W_{i,j}$ to be the projection of $W_i$ onto the $(j-i)$th block of coordinates.
We claim that
\begin{equation}\label{e:subspaceConstruction}
\begin{split}U_j &= \sum_{i=j-1}^{j-r-1} W_{i,j}\,,
\\
V_j &= \left(\bH{j,}{-} W_{j,j+1}
+ \sum_{i=j}^{j-r} \bH{j,}+ W_{i,j+1} \right)^\perp\,,
\end{split}\end{equation}
constitutes a feasible strategy for interference alignment.
Before rigorously justifying this, we first do a naive dimension count to verify
that
$$
\dim U_j = (r+1) \dim W_{i,j} = d
$$
and
$$\dim V_j = N - (r+2)\dim W_{i,j} = \frac{2r+3}{r+1} d - \frac{r+2}{r+1}d = d.
$$

Any element of $W_{i}$ consists of $r+1$ vectors $x_{i,j} \in
\CC^M$ for $i + 1 \leq j \leq i+r+1$, and these vectors satisfy $\bH{j,}+
x_{i,j+1} =  - \bH{j,}- x_{i, j+2}$ for $i+1 \leq j \leq i+r$.
First, since the channel matrices are injective, the only way for a
subvector $x_{i,j}$ to be zero is for the whole vector to be zero, and thus each
projection $W_{i,j}$ has the full dimension $d/(r+1)$. Second,
these equations
explain the apparent asymmetry in the definition of $V_j$, which can
equivalently be defined as the complement of the sum over all applications of
$\bH{j,}-$ and $\bH{j,}+$ to appropriate vector spaces $W_*$, but such vector
spaces coincide.  Indeed, this is the essence of the construction.
From this observation, it follows that that $\bH{j,}{+}U_{j+1}$ and $\bH{j,}{-}
U_{j-1}$ are orthogonal to $V_j$, which is what is required to be feasible.

The only thing remaining to be checked is that $U_j$ and $V_j$ actually have the
expected dimensions. This is verified in Lemma~\ref{l:UVdimension} below.

Finally, we suppose that $M=N = 2d$. The channel matrices are square, and thus,
generically, they are invertible, so we can define
\begin{equation*}
B = \bH{1,}2 (\bH{3,}2)^{-1} \bH{3,}1 (\bH{2,}1)^{-1} \bH{2,}3 (\bH{1,}3)^{-1}.
\end{equation*}
Again, generically, this matrix will have $2d$ distinct eigenvectors, and we
choose $V_1$ to be the span of any $d$ of them. Then we set
\begin{align*}
U_3^\perp &= (\bH{1,}3)^{-1} V_1 \\
V_2 &= \bH{2,}3 U_3^\perp \\
U_1^\perp &= (\bH{2,}1)^{-1} V_2 \\
V_3 &= \bH{3,}1 U_1^\perp \\
U_2^\perp &= (\bH{3,}2)^{-1} V_3.
\end{align*}
These form a feasible strategy.
\end{proof}

Note that our constructions imply that the alignment solution is unique when
$2d-M$ divides $d$, but there exist $2d \choose d$ solutions when $N=M=2d$.

\begin{lemma}\label{l:UVdimension}
The subspaces $U_j$ and $V_j$ defined in \eqref{e:subspaceConstruction} have dimension $d$. 
\end{lemma}
\begin{proof}
We first show that $U_1$ has dimension $d$; by symmetry of the construction, the dimensions of $U_2$ and $U_3$ will also be $d$.  

The subspace $U_1=\sum_{i=-r}^0 W_{i,1}$ is the sum of $r+1$ subspaces $W_{i,j}$, which we claim are independent;
suppose to the contrary, that there is some set of linearly dependent vectors $w_{i_1},w_{i_2},\dots, w_{i_s}$, with $0\leq i_1\leq i_2\leq\dots\leq i_s\leq r$, and $w_i\in W_{-i,1}$, satisfying $w_{i_s}-\sum_{\ell=1}^{s-1} \lam_\ell w_{i_\ell} =0$. Let $s$ be the minimum such value, with all sets of subspaces $W_{i_1,j},W_{i_2,j},\dots,W_{i_{s-1},j}$ for $j=1,2,3$ being complementary. 

Now, by the definition of the subspaces $W_{i,j}$, for each vector 
$w_{i_\ell}\in W_{-i_\ell,1}$ there is a sequence $u^{2}_{i_\ell},\dots,u^{q+1}_{i_\ell}$ of length $q:=r+1-i_{s-1}$ satisfying $\bH{3}{1}w_{i_\ell}=\bH{3}{2}u^{2}_{i_\ell},\dots, \bH{q+2,}{q}u^{q}_{i_\ell}=\bH{q+2,}{q+1}u^{q+1}_{i_\ell}$. The linear combination $\sum_{\ell=1}^{s-1}\lam_\ell w_{i_\ell}$ thus gives rise to a sequence $u^1,\dots, u^{q+1}$ defined by $u^a = \sum_{\ell=1}^{s-1}\lam_\ell u^a_{i_\ell}$ satisfying 
\begin{equation}\begin{split}\label{e:chain}
\bH{3}{1}w_{i_s}&=\bH{3}{1}\bigg(\sum_{\ell=1}^{s-1} \lam_\ell w_{i_\ell}\bigg) = \bH32 u^2,
\\
\bH12 u^2&=\bH13 u^3 \\
&\; \; \vdots
\\
\bH{q+2,}{q}u^{q}&=\bH{q+2,}{q+1}u^{q+1}\,.
\end{split}
\end{equation} 
Note that by the minimality assumption of $s$, none of the $u^j$ vectors are zero. 

By the definition of $W_{-i_s,1}$, there is a length-$(i_s-1)$ sequence of vectors preceding $w_{i_s}$ satisfying alignment conditions similar to those in \eqref{e:chain}; together with $w_{i_s}$ and the vectors in \eqref{e:chain}, this sequence can be extended to a sequence of vectors of total length $q+i_{s}=r+1+(i_s-i_{s-1})>r+1$, none of which are zero. Stacking the first $r+2$ of these vectors produces a nonzero element in the kernel of $\bA^{i_s}_{r+1}$. However, $\bA^{i_s}_{r+1}$ is full-rank by Lemma~\ref{l:AfullRank}; the dimension of the kernel is $\big[(r+2)M-(r+1)N\big]^+= M+d\big( 2r+1-2r-3\big)=M-2d<0$, i.e. the kernel is trivial. This is the desired contradiction. 

We now check that $V_1$ has dimension $d$, and again by symmetry, the dimensions of $V_2$ and $V_3$ will also be $d$. Note that if $V_1$ had dimension greater than $d$, we could choose a $d$-dimensional subspace and this would still satisfy the alignment equations \eqref{e:Generalorthogonality}. But $V_1$ is the orthogonal complement of the sum of $r+2$ subspaces $W_{i,j}$ of dimension $d/(r+1)$, so by subadditivity of dimension, we have the lower bound on dimension $\dim V_1\geq N-(r+2)\dim W_{i,j}=d$. 
\end{proof}

\section{Equation counting}
One requirement for feasibility of interference alignment is that the number of parameters defining the set of strategies must be at least the number of the constraints imposed by the decoding conditions. 
In \cite{BCT11}
and~\cite{RGL11}, this argument is carried out to obtain necessary
conditions on the feasibility of interference alignment, which in the setting of Theorem~\ref{t:SymDOF} reduces to the constraint
$$
4d\leq M+N\,.
$$ This constraint is strictly implied by the constraints~\eqref{e:feasibility} (see Fig.~\ref{f:feasibleRegion}).

Let us briefly recall the calculation leading to the inequality $4d\leq M+N$. The strategy space consists of a dimension $d$ subspace at each transmitter and receiver, i.e. a point in the product of Grassmannians  $\big(G(d,M)\big)^3 \times \big(G(d,N)\big)^3 $. This is an algebraic variety of dimension $3 d (M-d) + 3 d (N-d)$, while
the orthogonality conditions \eqref{e:Generalorthogonality} can be seen to impose $6d^2$ constraints. 
Comparing number of variables to number of constraints gives the inequality. 

\section*{Appendix: proof of Theorem~\ref{t:achievability}}
Here we prove Theorem~\ref{t:achievability} showing achievability for $M,N,d$ satisfying \eqref{e:feasibility}.
Let $r$ be the (unique) integer such that 
\begin{equation}
rN<(r+1)M\quad \text{and}\quad (r+1)N\geq (r+2)M\,.
\end{equation}
Note that this implies, from equation~\ref{e:feasibility}, that
\begin{equation}\label{e:N1}
(2r+3)d\leq (r+1)N
\end{equation}
and
\begin{equation}\label{e:M1} 
(2r+1)d\leq (r+1)M\,.
\end{equation}

We prove achievability by examining two cases: 1) $d\leq (r+1)[(r+1)M-rN]$ and 2) $d>(r+1)[(r+1)M-rN]$. Case 1 means that all of the signal space $U_i$ can be obtained from alignment paths of length $r+1$ (up to integer rounding), whereas in case 2 we must use alignment paths of length $r$ as well in order to attain the required $d$ dimensions.

We first assume case 1 holds. Consider $\bA_r^i$ as in the proof of Proposition~\ref{p:critical}, and let $W_i$ be a dimension $\lf \frac{d}{r+1}\rf$ subspace in the kernel of $\bA_r^i$. Let $d':=d-(r+1)\lf \frac{d}{r+1}\rf$, and if $d'>0$ let $w_i$ be a 1-dimensional subspace in $\ker \bA_r^i \setminus W_i$. The projections $W_{i,j}$ are defined in Proposition~\ref{p:critical} and the subspaces $w_{i,j}$ are defined analogously. The spaces $w_1,w_2,w_3$ are required in order to accommodate the remainder left when dividing $d$ by $r+1$, and will together contribute $d'$ dimensions to each signal space $U_j$. 
We put 
\begin{equation}
U_j = \sum_{i=j-1}^{j-r-1}W_{i,j} + \sum_{i=j-1}^{j-d'}w_{i,j}
\end{equation}
and 
\begin{equation}
\begin{split}
V_j &= \Bigg( \bH{j}{-}W_{j,j+1}+\bH{j}{-}w_{j,j+1}\\&\quad+\sum_{i=j}^{j-r}\bH{j}+ W_{i,j+1} + \sum_{i=j}^{j-d'+1}w_{i,j}\Bigg)^\perp.
\end{split}
\end{equation}

If all of $U_j$'s constituent subspaces are complementary, then $U_j$ has dimension $(r+1)\flr{\frac d{r+1}}+d'=d$; the justification for this statement is similar to the proof of Lemma~\ref{l:UVdimension} and omitted here. To see that $V_j$ has dimension (at least) $d$, we observe that by subadditivity of dimension,
\begin{equation}
\dim V_j\geq N-(r+2)\left\lf \frac{d}{r+1}\right\rf-d' -e\,, 
\end{equation} where $e=0$ if $(r+1)|d$ and $e=1$ otherwise. 
Plugging in the inequality \eqref{e:N1} we obtain
\begin{align*}
\dim V_j &\geq \frac{2r+3}{r+1}d-d-e-\left\lf \frac{d}{r+1}\right\rf \\
&=d+\frac{d}{r+1}-\left\lf \frac{d}{r+1}\right\rf-e\geq d\,.
\end{align*}

Suppose now that case 2 holds, i.e. $d>(r+1)[(r+1)M-rN]$. This means that not all of the signal space~$U_i$ can be included in alignment paths of length $r+1$, so the remainder will be included in alignment paths of lenth $r$. Let $d':=d-(r+1)[(r+1)M-rN]$ and $d'' = d'-r\flr{\frac{d'}r}$. 
As before, denote by $W_i$ the kernel of the matrix $\bA_r^i$, having dimension $(r+1)M-rN$. Denote by $\pi$ the projection from $\C^{(r+1)M}\to\C^{rM}$ to the first $rM$ coordinates. The space $\pi(\ker \bA_r^i)$ is contained in $\bA_{r-1}^i$. Let $X_i$ for $i=1,2,3$ each be a $\flr{\frac{d'}{r}}$ dimensional subspace in $\ker \bA_{r-1}^i\setminus \pi(W_i)$, and let $w_i$ be a 1-dimensional subspace in $\ker \bA_{r-1}^i\setminus (\pi(W_i)+ X_i)$. 
Put
\begin{equation}
U_j = \sum_{i=j-1}^{j-r-1}W_{i,j} + \sum_{i=j-1}^{j-r}X_{i,j} +\sum_{i=j-1}^{j-d''}w_{i,j}
\end{equation}
and 
\begin{equation}
\begin{split}
&V_j = \Bigg( \bH{j}{-}(W_{j,j+1}+X_{j,j+1}+w_{j,j+1}) \\& \!+\sum_{i=j}^{j-r}\bH{j}\!+ W_{i,j+1}+\sum_{i=j}^{j-r+1}\bH{j}+ X_{i,j+1} \!+ \sum_{i=j}^{j-d'+1}w_{i,j}\Bigg)^\perp\!\!.
\end{split}
\end{equation}

As before, a naive count suggests that $U_j$ should have dimension $d$, and this can be justified similarly to Lemma~\ref{l:UVdimension}.

To see that $V_j$ has dimension at least $d$ we again use subadditivity of dimension to get 
\begin{align*}
&\dim V_j\\&\geq N-(r+2)[(r+1)M-rN]-(r+1)\flr{\frac{d'}r}-d''-e_1 \\
&= N-(r+2)[(r+1)M-rN]- \flr{\frac{d'}r}-d'-e_1\,,
\end{align*}
where $e_1$ is zero if $r|d'$ and $e_1$ is one otherwise. Letting $e_2:=\frac{d'}r - \flr{\frac{d'}r}$, we 
have 
\begin{align*}
\dim V_j&\geq N-(r+2)[(r+1)M-rN]-\frac{d'}r-d'+e_2-e_1 \\
&=N-\frac{(r+1)d}{r}+\frac{1}r[(r+1)M-rN]+e_2-e_1\\
&=d+\frac{(r+1)}{r}M-\frac{2r+1}{r}d+e_2-e_1
\,.
\end{align*}
Substituting $\frac{r+1}{2r+1}M$ for $d$, the inequality \eqref{e:M1} implies that
$$
\dim V_j\geq d+e_2-e_1\,.
$$
If $e_1$ is one then $e_2$ is strictly positive, so the fact that $\dim V_j$ is an integer implies $\dim V_j\geq d$. 

\bibliographystyle{ieeetr}
\bibliography{BIB}

\end{document}